\documentclass[11pt,letterpaper]{article}

\usepackage{authblk,amsmath,amsfonts,amssymb,amsthm,listings,fullpage}
\date{}
\theoremstyle{plain}
\newtheorem{theorem}{Theorem}
\newtheorem{proposition}[theorem]{Proposition}
\newtheorem{fact}[theorem]{Fact}
\newtheorem{lemma}[theorem]{Lemma}
\newtheorem{definition}[theorem]{Definition}

\usepackage[utf8]{inputenc}

\usepackage{tikz}
\usetikzlibrary {arrows,snakes,backgrounds,patterns}
\usepackage[caps]{complexity}

\newcommand{\PQ}{\lang{PQ}}
\newcommand{\PQT}{\lang{PQ-TS}}
\newcommand{\PQs}{\PQ}
\newcommand{\PQTs}{\PQT}
\newcommand{\HAI}{\lang{WeakIndex}}
\newcommand\MHAI[2]{\lang{WeakIndex}(#1,#2)}
\newcommand\RD[2]{\lang{Raindrops}(#1,#2)}
\DeclareMathOperator*{\Exp}{{\mathbb E}}
\DeclareMathOperator{\hash}{hash}
\DeclareMathOperator{\update}{\texttt{Update}}

\newcommand{\hel}{\mathrm{h}}
\newcommand{\mi}{\mathrm{I}}
\newcommand{\h}{\mathrm{H}}
\newcommand{\Order}{\mathrm{O}}
\newcommand{\eps}{\varepsilon}
\newcommand{\ins}{\texttt{ins}}
\newcommand{\ext}{\texttt{ext}}
\newcommand{\Col}{\textsc{Collection}}
\newcommand{\N}{{\mathbb N}}
\renewcommand{\polylog}{\mathrm{polylog}}

\newcommand\restr[2]{
  \left.\kern-\nulldelimiterspace 
  #1 
  \vphantom{\big|} 
  \right|_{#2} 
  }

\newcommand{\comments}[1]{}

\title{Streaming Complexity of Checking \mbox{Priority Queues}%
\footnote{Supported by the French ANR Defis program under contract ANR-08-EMER-012 (QRAC project)}
}

\author[1]{Nathana\"el Fran\c{c}ois}
\author[2]{Fr\'ed\'eric Magniez}
\affil[1]{Univ Paris Diderot, Sorbonne Paris-Cit\'e, LIAFA, CNRS, 75205 Paris, France\\
\texttt{nathanael.francois@liafa.univ-paris-diderot.fr}}
\affil[2]{CNRS, LIAFA, Univ Paris Diderot, Sorbonne Paris-Cit\'e, 75205 Paris, France\\
\texttt{frederic.magniez@univ-paris-diderot.fr}
}

\begin{document}

\maketitle

\begin{abstract}
This work is in the line of designing efficient checkers for testing the reliability of
some massive data structures.
Given a sequential access to the insert/extract operations on such a structure,
one would like to decide, a posteriori only, if it corresponds to the evolution of a reliable structure.
In a context of massive data, one would like to minimize both the amount of reliable memory of the checker
and the number of passes on the sequence of operations.

Chu, Kannan and McGregor~\cite{ckm07} initiated the study of checking priority queues in this setting.
They showed that the use of timestamps allows to check a priority queue with a single pass and memory space
$\tilde{\Order}(\sqrt{N})$. Later, Chakrabarti, Cormode, Kondapally and McGregor~\cite{cckm10} removed
the use of timestamps, and proved that more passes do not help.

We show that, even in the presence of timestamps, more passes do not help, solving an open problem
of~\cite{ckm07,cckm10}. On the other hand, we show that a second pass, but in {\em reverse} direction,
shrinks the memory space to $\tilde{\Order}((\log N)^2)$, extending a phenomenon the first time
observed by Magniez, Mathieu and Nayak~\cite{mmn10} for checking well-parenthesized expressions.
\end{abstract}

\section{Introduction}
The reliability of memory is central and becomes challenging when it is massive.
In the context of program checking~\cite{bk95} this problem has been
addressed by Blum, Evans, Gemmell, Kannan and Naor~\cite{begkn94}. 
They designed on-line checkers that use a small amount of reliable
memory to test the behavior of some data structures. 
Checkers are allowed to be randomized and to err with small error probability.
 In that case the error probability is not over the inputs but over the random coins of the algorithm.

Chu, Kannan and McGregor~\cite{ckm07} revisited this problem for priority queue data structures, 
where the checker only has to detect an error after processing an entire sequence of data accesses.
This can be rephrased as a one-pass streaming recognition problem.
Streaming algorithms sequentially scan the whole input piece by piece in one sequential pass,
or in a small number of passes,
while using sublinear memory space. 
In our context, the stream is defined by the sequence of insertions and extractions on the priority queue.
Using a streaming algorithm, the objective is then 
to decide if the stream corresponds to a correct implementation of a priority queue.
We also consider collection data structures that implement multisets. 
\begin{definition}[\Col,\PQ]
Let $\Sigma_0$ be some alphabet.
Let $\Sigma=\{\ins(a),\ext(a): a \in \Sigma_0\}$. For $w \in \Sigma^N$,
define inductively multisets $M_i$
by  $M_0=\emptyset$, $M_i=M_{i-1}\setminus \{a\}$ if $w[i]=\ext(a)$,
and $M_i=M_{i-1} \cup \{a\}$ if $w[i]=\ins(a)$.\\
Then $w\in \Col(\Sigma_0)$ if and only if $M_n=\emptyset$
and $a\in M_{i-1}$ when $w[i]=\ext(a)$, for $i=1,\dots,N$.
Moreover, $w\in \PQ(U)$, for $U\in\N$,  if and only if $w\in\Col(\{0,1,\ldots,U\})$ and
$a=\max(M_{i-1})$ when $w[i]=\ext(a)$, for $i=1,\dots,N$.
\end{definition}

Streaming algorithms were initially designed with a single pass: 
when a piece of the stream has been read, it is gone for ever.
This makes those algorithms of practical interest for online context, such as network monitoring,
for which first streaming algorithms were developed~\cite{ams99}.
Motivated by the explosion in the size of the data that algorithms are
called upon to process in everyday real-time applications, 
the area of streaming algorithms has experienced tremendous growth over the last decade in many applications.
In particular, a streaming algorithm can model an external read-only memory. 
Examples of such applications occur in bioinformatics for genome
decoding, or in Web databases for the search of documents. 
In that context, considering multi-pass streaming algorithm is relevant.

Using standard arguments one can establish that every $p$-pass randomized streaming algorithm
needs memory space $\Omega(N/p)$ for recognizing $\Col$.
Nonetheless, 
Chakrabarti, Cormode, Kondapally and McGregor~\cite{cckm10} gave a one-pass randomized for $\PQ$ using memory space
$\tilde{O}(\sqrt{N})$.
They also showed that several passes do not help, since any $p$-pass
randomized algorithm would require memory space $\Omega({\sqrt{N}}/{p})$. 
A similar lower bound was showed independently, but using different tools,
by Jain and Nayak~\cite{jn10}. 
The case of a single pass was established previously by Magniez, Mathieu and Nayak~\cite{mmn10}
for checking the well-formedness of parenthesis expressions, or equivalently the behavior of a stack.

A simpler variant of $\PQ$ with timestamps was in fact first studied by Chu, Kannan and McGregor~\cite{ckm07},
where now each item is inserted to the queue with its index. 

\begin{definition}[\PQT]
Let $\Sigma=\{\ins(a),\ext(a): a \in \{0,1,\ldots,U\}\}\times {\mathbb
  N}$. Let $w \in \Sigma^N$. 
Then $w\in\PQT(U)$ if and only if 
$w\in\Col(\Sigma)$, $w[1,\dots,N][1]\in \PQ(U)$,
and $w[i][2]=i$ when $w[i][1]=\ins(a)$.
\end{definition}

Nonetheless the two works~\cite{ckm07,cckm10} let open two problems.
The lower bound of~\cite{cckm10} was only proved for $\PQ$, and
no significant lower bounds for $\PQT$ was established.
Moreover, the streaming complexity of $\PQ$ for algorithms that can
process the stream in any direction has not been studied. 

Even though recognizing $\PQT$ is obviously easier than recognizing $\PQ$, 
our first contribution (Section~\ref{sec:lb}) consists in showing that
they both obey the same limitation, even with multiple passes in the
same direction. 
\begin{theorem}\label{main1}
Every $p$-pass randomized streaming algorithm 
recognizing $\PQT(3N/2)$ with bounded error $1/3$   
requires memory space $\Omega({\sqrt{N}}/{p})$ for inputs  of length $N$.
\label{multilowerbound}
\end{theorem}
As a consequence, since this lower bound uses very restricted hard
instances, it models most of possible variations.  
For instance, assuming that the input is in $\Col$ and has no
duplicates, is not sufficient to guarantee a faster algorithm. 
The proof of Theorem~\ref{main1} consists in introducing a related
communication problem with $\Theta(\sqrt{N})$ players. 
Then we reduce the number of players to $3$, and prove a lower bound
on the information carried by players, leading to the desired lower
bound. 
We are following the {\em information cost} approach taken 
in~\cite{ChakrabartiSWY01,SaksS02,Bar-YossefJKS04,JayramKS03,JainRS03b},
among other works. 
Recently, the information cost appeared as one of the most central notion
in communication complexity~\cite{br11,bra12,kllrx12}.
The information cost of a protocol 
is the amount of information that messages carry about players' inputs.
We adapt this notion to suit both the nature of streaming algorithms and of our problem.

Even if our result suggests that allowing multiple passes does not help,
one could also consider the case of bidirectional passes.
We believe that it is a natural relaxation of multi-pass streaming
algorithms where the stream models some external read-only memory.
In that case, we show that a second pass, but in reverse order, makes the problem of checking $\PQ$ easy, 
even with no timestamps (Section~\ref{sec:algo}).
A similar phenomenon has been established previously in~\cite{mmn10} for checking the well-formedness
of parenthesis expressions. Their problem is simpler than ours, and therefore our algorithm is more general.
\begin{theorem}\label{main2}
There is a bidirectional $2$-pass randomized streaming algorithm
recognizing $\PQ(U)$ 
with memory space $\Order((\log N)(\log U+\log N))$, 
time per processing item $\polylog(N,U)$, and one-sided bounded error ${N^{-c}}$,
for inputs  of length $N$ and any constant $c>0$.
\end{theorem}

Our algorithm uses a hierarchical data structure similar to the one
introduced in~\cite{mmn10} for checking well-parenthesized
expressions. At high level, it also behaves similarly. 
It performs one pass in each direction and makes an on-line
compression of past information in at most $\log N$ hashcodes.
While this compression can loose information, the compression technique ensures that 
a mistake is always detected in one of the two directions. Nonetheless our algorithm differs on two main points.
First, unlike parenthesized expressions, $\PQ$ is not symmetric. 
Therefore one has to design an algorithm for each pass.
Second, the one-pass algorithm for $\PQ$~\cite{cckm10}
is technically more advanced than the one of~\cite{mmn10}. Thus designing a bidirectional
$2$-pass algorithm for $\PQ$ is more challenging.

Theorems~\ref{main1} and~\ref{main2} point out a strange situation but not isolated at all.
Languages studied in~\cite{ckm07,mmn10,cckm10,km12} and in this paper
have space complexity $\Theta(\sqrt{N}\polylog(N))$
for a single pass, $\Omega(\sqrt{N}/p)$ for $p$ passes in the same direction, 
and $\polylog(N)$ for $2$ passes but one in each direction.
We hope this paper makes progress in the study that phenomenon.

\section{Preliminaries}

In streaming algorithms (see~\cite{MuthuBook} for an introduction), a {\em pass\/}
on an input $w\in\Sigma^N$, for some alphabet $\Sigma$, means that $w$
is given as an {\em input stream} $w[1],w[2],\ldots,w[N]$, which
arrives sequentially, i.e., letter by letter in this order.
For simplicity, we assume throughout this article that the input length $N$
is always given to the algorithm in advance. Nonetheless,
all our algorithms can be adapted to the case in which $N$ is
unknown until the end of a pass.  
\begin{definition}[Streaming algorithm]
   A $p$-pass randomized
   {\em streaming algorithm}  with space
  $s(N)$ and time $t(N)$ is a randomized
  algorithm that, given $w\in\Sigma^N$ as an input stream,
\begin{itemize}
\item  performs $k$ sequential passes on $w$;
\item maintains a memory space of size at most $s(N)$ bits
while reading $w$;
\item has running time at most $t(N)$ per processed letter $w[i]$;
\item has preprocessing and postprocessing time at most $t(N)$.
\end{itemize}
The algorithm is {\em bidirectional\/} if it is allowed to
access to the input in the reverse order, after reaching the end of
the input. Then $p$ is the total number of passes in
either direction.
\end{definition}

The proof of our lower bound uses the language of communication
complexity with multi-players, 
and is based on information theory arguments.
We consider {\em number-in-hand} and  {\em message-passing} communication protocols.
Each player is given some input, and can communicate with another
player according to the rules of the protocol.
Our players are embedded into a directed circle, so that each
player can receive (resp. transmit) a message from its unique
predecessor (resp. successor). Each player send a message after
receiving one, until the end of the protocol is reached. 
Players have no space and time restriction. Only the number of rounds
and the size of messages are constrained.

Consider a randomized multi-player communication protocol $P$.
We consider only two types of random source, that we call {\em coins}.
Each player has access to its own independent source of {\em private coins}.
In addition, all players share another common source of {\em public coins}.
The output of $P$ is announced by the last player. This is therefore
the last message of the last player. We say that $P$ is with bounded
error $\epsilon$ when $P$ errs with probability at most $\eps$ over
the private and public coins. 
The {\em transcript $\Pi$} of $P$ is the concatenation of all messages
sent by all players, including all public coins. In particular, it
contains the output of $P$, since it is given by the last player. 
Given a subset $S$ of players, we let $\Pi_S$ be the concatenation of
all messages  sent by players in $S$, including again all public
coins.

We now remind the usual notions of entropy $\h$ and mutual information $\mi$.
Let $X,Y,Z$ be random variables. Then $\h(X)=-\Exp_{x\gets X}\log \Pr(X=x)$,
$\h(X|Y=y)=-\Exp_{y\gets Y}\log \Pr(X=x|Y=y)$,
$\h(X|Y)=\Exp_{y\gets Y} \h(X|Y=y)$,
and $\mi(X:Y|Z)=\h(X|Z)-\h(X|Y,Z)$. 
The entropy and the mutual information are non negative
and satisfy $\mi(X:Y|Z)=\mi(Y:X|Z)$.

The mutual information between two random variables is connected to the Hellinger distance $\hel$
between their respective distribution probabilities. Given a random variable $X$ we also
denote by $X$ its underlying distribution.
\begin{proposition}[Average encoding]
Let $X,Y$ be random variables. Then
$\Exp_{y \gets Y}\hel^2(X|_{Y=y},X) \leq \kappa\mi(X:Y)$, where
$\kappa=\frac{\ln 2}{2}$.
\end{proposition}

The Hellinger distance also generalizes the cut-and-paste property of deterministic protocols
to randomized ones.
\begin{proposition}[Cut and paste]
Let $P$ be a $2$-player randomized  protocol. Let
$\Pi(x,y)$ denote the random variable representing the transcript in
$P$ when Players $A,B$ have resp. inputs $x,y$. Then $\hel(\Pi(x,y),\Pi(u,v)) = \hel(\Pi(x,v),\Pi(u,y))$,
for all pairs  $(x,y)$ and $(u,v)$.
\label{cut-and-paste}
\end{proposition}

Last we use that the square of the Hellinger distance is convex, and the following
connexion to the more convention $\ell_1$-distance:
$
\hel(X,Y)^2 \leq \frac{1}{2}\lVert X-Y \rVert_1 \leq \sqrt{2} \hel(X,Y)
$.
For a reference on these results, see~\cite{jn10}.

\section{Lower bound for $\PQT$}\label{sec:lb}
The proof of our lower bound consists in first translating it into
a $3m$-player communication problem, for some large $m$;
then reducing the number of players to $3$ using the information cost approach;
and last studying the base case of $3$ players using information theory arguments.

\subsection{From streaming algorithms to communication protocols}
In this section, we write $a$ instead of $\ins(a)$ and $\bar{a}$ instead of $\ext(a)$.
Consider the following set of hard instances of size $N=(2n+2)m$:
\begin{quote}
$\RD{m}{n}$  (see LHS of Figure~\ref{hard_instance})
\begin{itemize}
\item For $i=1,2,\dots,m$, repeat the following motif:
\begin{itemize}
\item For $j=1,2,\dots,n$, insert either $v_{i,j}=3(ni-j)$ or
  $v_{i,j}=3(ni-j)+2$
\item Insert either $a_i=3(ni-(k_i-1))+1$  or $a_i=3(ni-k_i)+1$, for some $k_i \in \{2,\dots,n\}$
\item Extract $v_{i,1},v_{i,2},\dots,v_{i,k_i-1},a_i$ in
  decreasing order
\end{itemize}
\item Extract everything left in decreasing order
\end{itemize}
\end{quote}

Observe that such an instance is in $\Col$. 
One can compute the timestamps for each value by maintaining only
$\Order(\log N)$ additionnal bits.
Last, there is only one potential error in each motif
that can make it outside of $\PQTs$. 
Indeed, $v_{i,1},v_{i,2},\dots,v_{i,k_i-1},a_i$ are in decreasing order up to a switch between $a_i$ and $v_{i,k_i-1}$.

\begin{figure}
\hspace*{1cm}{
\begin{tikzpicture}[thick,scale=0.52, every node/.style={scale=0.7}]
\node at (0.75,-0.5) {\phantom{$A_1$}};
\draw [->,thick] (-0.5,-0.5) -- (-0.5,12.5);
\draw [->,thick] (-1,0) -- (11,0);

\draw[xstep=0.5,ystep=0.5,gray!40,very thin] (-0.7,-0.2) grid
(7.9,12.4);
\draw[xstep=100,ystep=0.5,gray!40,very thin,dashed] (8,-0.2) grid
(8.5,12.4);
\draw[xstep=0.5,ystep=0.5,gray!40,very thin] (8.6,-0.2) grid (10.9,12.4);

\node  at (1.5,1) {$2$};
\node  at (1,2.5) {$5$};
\node  at (0.5,4) {$8$};
\node  at (0,4.5) {$9$};
\node[draw,circle]  at (2,3.5) {$7$};
\node  at (2.5,4.5) {$\overline{9}$};
\node  at (3,4) {$\overline{8}$};
\node  at (3.5,3.5) {$\overline{7}$};
\node  at (5.5,7) {$14$};
\node  at (5,8.5) {$17$};
\node  at (4.5,9) {$18$};
\node  at (4,11.5) {$23$};
\node[draw,circle]  at (6,8) {$16$};
\node  at (6.5,11.5) {$\overline{23}$};
\node  at (7,9) {$\overline{18}$};
\node  at (7.5,8) {$\mathbf{\overline{16}}$};

\node at (9,8.5) {$\mathbf{\overline{17}}$};
\node at (9.5,7) {$\overline{14}$};
\node at (10,2.5) {$\overline{5}$};
\node at (10.5,1) {$\overline{2}$};

\draw (-0.7,5.75) -- (7.9,5.75);
\draw[dashed] (8,5.75) -- (8.5,5.75);
\draw (8.6,5.75) -- (11,5.75);
\draw (-0.7,11.85) -- (7.9,11.85);
\draw[dashed] (8,11.85) -- (8.5,11.85);
\draw (8.6,11.85) -- (11,11.85);

\node at (-1.5,3) {$i=1$};
\node at (-1.5,2) {$k=3$};
\node at (-1.5,9) {$i=2$};
\node at (-1.5,8) {$k=3$};
\node at (-1.5,12.5) {$i=3$};

\end{tikzpicture}
}
{
\begin{tikzpicture}[thick,scale=0.52, every node/.style={scale=0.7}]

\fill[fill=gray!20,fill opacity = 0.5] (-0.5,0) -- (-0.5,5.75) --
(1.75,5.75) -- (1.75,0); 
\fill[fill=gray!50,fill opacity = 0.5] (3.75,0) -- (3.75,5.75) --
(1.75,5.75) -- (1.75,0);
\fill[fill=gray!20,fill opacity = 0.5] (3.75,11.85) -- (3.75,5.75) --
(5.75,5.75) -- (5.75,11.85);
\fill[fill=gray!50,fill opacity = 0.5] (5.75,11.85) -- (5.75,5.75) --
(8,5.75) -- (8,11.85);
\fill[fill=gray!50,fill opacity = 0.5] (8.5,5.75) -- (8.5,11.85) --
(9.75,11.85) -- (9.75,5.75);
\fill[fill=gray!50,fill opacity = 0.5] (10.75,0) -- (10.75,5.75) --
(9.75,5.75) -- (9.75,0);

\draw [->,thick] (-0.5,-0.5) -- (-0.5,12.5);
\draw [->,thick] (-1,0) -- (11,0);

\draw[xstep=0.5,ystep=0.5,gray!60,very thin] (-0.7,-0.2) grid
(7.9,12.4);
\draw[xstep=100,ystep=0.5,gray!60,very thin,dashed] (8,-0.2) grid
(8.5,12.4);
\draw[xstep=0.5,ystep=0.5,gray!60,very thin] (8.6,-0.2) grid (10.9,12.4);

\node  at (1.5,1) {$2$};
\node  at (1,2.5) {$5$};
\node  at (0.5,4) {$8$};
\node  at (0,4.5) {$9$};
\node  at (2,3.5) {$7$};
\node  at (2.5,4.5) {$\overline{9}$};
\node  at (3,4) {$\overline{8}$};
\node  at (3.5,3.5) {$\overline{7}$};
\node  at (5.5,7) {$14$};
\node  at (5,8.5) {$17$};
\node  at (4.5,9) {$18$};
\node  at (4,11.5) {$23$};
\node  at (6,8) {$16$};
\node  at (6.5,11.5) {$\overline{23}$};
\node  at (7,9) {$\overline{18}$};
\node  at (7.5,8) {$\mathbf{\overline{16}}$};

\node at (9,8.5) {$\mathbf{\overline{17}}$};
\node at (9.5,7) {$\overline{14}$};
\node at (10,2.5) {$\overline{5}$};
\node at (10.5,1) {$\overline{2}$};

\draw (-0.7,5.75) -- (7.9,5.75);
\draw[dashed] (8,5.75) -- (8.5,5.75);
\draw (8.6,5.75) -- (11,5.75);
\draw (-0.7,11.85) -- (7.9,11.85);
\draw[dashed] (8,11.85) -- (8.5,11.85);
\draw (8.6,11.85) -- (11,11.85);

\node at (-1.5,2.5) {$i=1$};
\node at (-1.5,8.5) {$i=2$};
\node at (-1.5,12.5) {$i=3$};
\node at (0.75,-0.5) {$A_1$};
\node at (2.75,-0.5) {$B_1$};
\node at (4.75,5.25) {$A_2$};
\node at (6.75,5.25) {$B_2$};
\node at (9.25,5.25) {$C_2$};
\node at (10.25,-0.5) {$C_1$};

\end{tikzpicture}
}\vspace*{-5mm}
\caption{
Left: \label{hard_instance} Instance of $\RD{m}{4}$  with one error:
17 is extracted after 16. Insertions $a_i$ are circled.
Right: \label{hard_instance_2}Cutting $\RD{m}{4}$ into $3m$ pieces to
make it a communication problem. Players' input are within each corresponding region.}
\end{figure}
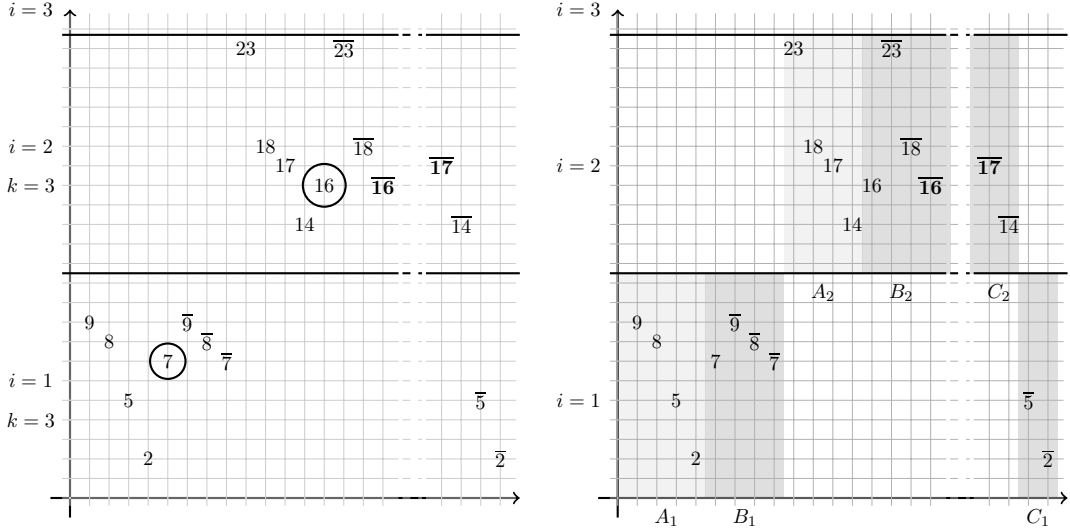

Given such an instance as a stream, an algorithm for $\PQTs$ must decide
if an error occurs between $\overline{a_i}$ and $\overline{v_{i,k_i}}$, for some $i$.
Intuitively, if the memory space is less than $\eps n$, for a small enough constant $\eps>0$, then
the algorithm cannot remember all the values $(v_{i,j})_j$ when $a_i$ is extracted,
and therefore cannot check a potential error with $a_i$.
The next opportunity is during the last sequence of extractions.
But then, the algorithm has to remember all values $(a_i)_i$, which is again impossible
if the memory space is less than $\eps m$.

In order to formalize this intuition, Lemma~\ref{stream-comm} (proof in Appendix~\ref{app:lb})
first translates our problem into a communication one between $3m$ players
as shown on the RHS of Figure~\ref{hard_instance_2}. Then we analyze its complexity
using information theory arguments in Section~\ref{lb:com}.

Any insertion and extraction of an instance in $\RD{m}{n}$ can be described by its index and a single bit. 
Let $x_i[j] \in\{0,1\}$ such that $v_{i,j} = 3(ni-j) +
2x_i[j]$. Similarly, let $d_i \in \{0,1\}$ such that
$a_i=3(ni-k_i)+1+3d_i$. 
For simplicity, we write $\mathbf{x}$ instead of $(x_i)_{1 \leq i \leq
  m}$. Similarly, we use the notations $\mathbf{k}$ and $\mathbf{d}$. 
Then our related communication problem is:
\begin{quote}
$\HAI(m,n)$
\begin{itemize}
\item Input for players $(A_i,B_i,C_i)_{1\leq i\leq m}$:
\begin{itemize}
\item Player $A_i$ has a sequence $x_i \in \{0,1\}^n$
\item Player $B_i$ has $x_i[1,k_i-1]$, with $k_i \in \{2,\dots,n\}$ and $d_i \in \{0,1\}$
\item Player $C_i$ has $x_i[k_i,n]$
\end{itemize}
\item Output: $f_m(\mathbf{x},\mathbf{k},\mathbf{d})=
\bigvee_{i=1}^{m}f(x_i,k_i,d_i)$, where $f(x,k,d)=[(d=0) \wedge (x[k]=1)]$
\item Communication settings:
\begin{itemize}
\item One round: each player  sends a message to the next player
  according to the diagram $A_1 \rightarrow B_1 \rightarrow A_2 
  \rightarrow \dots \rightarrow B_m \rightarrow C_m \rightarrow C_{m-1}
  \rightarrow \dots \rightarrow C_1$.
\item Multiple rounds: If there is at least one round left, $C_1$
  sends a message to $A_1$, and then players continue with the next
  round. 
\end{itemize}
\end{itemize}
\end{quote}

\begin{lemma}
Assume there is a  $p$-pass randomized streaming algorithm 
for deciding if an instance of $\RD{n}{m}$ is in $\PQTs(3mn)$
with memory space $s(m,n)$ and bounded error $\eps$.
Then there is a  $p$-round randomized protocol for $\MHAI{n}{m}$ with bounded error $\eps$ 
such that each message has size at most $s(m,n)$.
\label{stream-comm}
\end{lemma}

We are now ready to give the structure of the proof of
Theorem~\ref{multilowerbound}, which has techniques based on
information theory. Define the following collapsing distribution $\mu_0$ of hard inputs $(x,k,d)$, 
encoding instances of $\RD{1}{n}$, where $f$ always takes value~$0$. 
Distribution $\mu_0$ is such that $(x,k)$ 
is uniform on $\{0,1\}^n \times \{2,\dots,n\}$ and, given $x,k$, the
bit $d\in\{0,1\}$ is uniform if $x[k]=0$, and $d=1$ if $x[k]=1$.
From now on, $(X,K,D)$ are random variables distributed according to $\mu_0$,
and $(x,k,d)$ denote any of their values.

Then the proof of Theorem~\ref{multilowerbound} consists
in studying the information cost   of any communication protocol for $\MHAI{n}{m}$, which is a lower bound on
its communication complexity. 
Using that $\mu_0$ is collapsing for $f$, Lemma~\ref{multi_dir_sum}
establishes a direct sum on the information cost of $\MHAI{n}{m}$.
Then, even if $f$ is constant on $\mu_0$,
Lemma~\ref{lower_bound_epsilon} lower bounds the information cost of a single instance of $\MHAI{n}{1}$.
\begin{proof}[Proof of Theorem~\ref{multilowerbound}]
Let $n,N$ be positive integers such that $N=(2n+2)n$.
Assume that there exists a  $p$-pass randomized algorithm that recognizes $\PQTs(3N/2)$,
with memory space $\alpha n$ and bounded error $\eps$, for inputs of size $N$.
Then, by Lemma~\ref{stream-comm}, there a 
 $p$-round randomized protocol $P$ for $\MHAI{n}{n}$ such that each
message has size at most $\alpha n$. By Lemma~\ref{multi_dir_sum},
one can derive from $P$ another  $(p+1)$-round randomized protocol $P'$ for $\MHAI{n}{1}$
with bounded error $\eps$, and transcript $\Pi'$ satisfying
$|\Pi'|\leq 3(t+1)\alpha n$ and $\max\left\{\mi(D:\Pi'_B|X,K),
\mi(K,D:\Pi'_C|X)\right\} \leq (p+1)\alpha$.
Then by Lemma~\ref{lower_bound_epsilon},
$3(p+1)\alpha \geq (1-2\eps)/10$,
that is $\alpha=\Order(1/p)$, concluding the proof.
\end{proof}

\subsection{Communication complexity lower bound}\label{lb:com}

We first reduce the general problem $\MHAI{n}{m}$ with $3m$ players
to a single instance of $\MHAI{n}{1}$ with 3 players.
In order to do so we exploit the direct sum property of the information cost.
The use of a collapsing distribution where $f$ is always $0$ is crucial.
\begin{lemma}
If there is a $p$-round randomized  protocol $P$
for $\MHAI{n}{m}$ with bounded error $\eps$ and messages of
size at most $s(m,n)$, then there is a $(p+1)$-round
randomized  protocol $P'$ for $\MHAI{n}{1}$ with bounded
error $\epsilon$, and transcript $P'$ satisfying
$ |\Pi'| \leq 3(p+1)s(m,n)$ and
$\max \left\{ \mi(D:\Pi'_B|X,K),
\mi(K,D:\Pi'_C|X) \right\} \leq \frac{p+1}{m}s(m,n)$.
\label{multi_dir_sum}
\end{lemma}
\begin{proof}[Sketch of proof]
Given a protocol $P$, we show how to construct another protocol $P'$ for any instance $(x,k,d)$ of $\MHAI{n}{1}$.
In order to avoid any confusion, we denote by $A$, $B$ and $C$ the three players of $P'$, 
and by $(A_i,B_i,C_i)_i$ the ones of $P$.
\begin{quote}
Protocol $P'$
\begin{itemize}
\item Using public coins, all players generate uniformly at random $j \in \{1,\dots,m\}$,
and $x_i\in\{0,1\}^n$ for $i \not = j$
\item Players $A$, $B$ and $C$ set respectively their inputs to the ones of $A_j, B_j, C_j$
\item For all $i>j$, Player $B$ generates, using its private coins,
 uniformly at random $k_i\in\{2,\ldots,n\}$,
  and then it generates uniformly at random $d_i$ such that $f(x_i,k_i,d_i)=0$
\item For all $i<j$, Player $C$ generates, using its private coins,
 uniformly at random $k_i\in\{2,\ldots,n\}$,
  and then it generates uniformly at random $d_i$ such that $f(x_i,k_i,d_i)=0$
\item Players $A$, $B$ and $C$ run $P$ as follows.
$A$ simulates $A_j$ only, $B$ simulates $B_j$ and $(A_i,B_i,C_i)_{i>j}$, and
$C$ simulates $C_j$ and $(A_i,B_i,C_i)_{i<j}$.
\end{itemize} 
\end{quote}
Observe that $A$ starts the protocol if $j=1$, and $C$ starts otherwise.
Moreover $C$ stops the simulation after $p$ rounds if $j=1$, and after $p+1$ rounds otherwise.
For all $i \not = j$, entries are generated such that $f(x_i,k_i,a_i)=0$, therefore
$f_m(\mathbf{X},\mathbf{k},\mathbf{d})=f(x_j,k_j,a_j)=f(x,k,a)$,
and $P'$ has the same bounded error than $P$.

Then we show in Appendix~\ref{app:lb} that $P'$ satisfies the required conditions of the lemma.
\end{proof}

We now prove a trade-off between the bounded error of a protocol for
a single instance of $\MHAI{n}{1}$ and its information cost.
The proof involves some of the tools of~\cite{jn10} but with some additional obstacles to apply them.
The inherent difficulty is due to that we have $3$ players whereas
the cute-and-paste property applies to $2$-player protocols.
Therefore we have to group $2$ players together.

Given some parameters $(x,k,a)$ for an input of $\MHAI{n}{1}$, we
denote by $\Pi(x,k,a)$ the random variable describing the transcript $\Pi$ of our protocol.
We start by two lemmas exploiting the average encoding theorem (proofs in Appendix~\ref{app:lb}).

\begin{lemma}
Let $P$ be a randomized  protocol for $\MHAI{n}{1}$ 
with transcript $\Pi$ satisfying
$ |\Pi| \leq \alpha n$ and
$\mi(K,D:\Pi_C|X) \leq \alpha$.
Then
$$\Exp_{x[1,l-1],l} \hel^2(\Pi(x[1,l-1]0X[l+1,n],l,1),  \Pi(x[1,l-1]1X[l+1,n],l,1))
\leq 28\alpha,$$
where $l\in[\frac{n}{2}+1,n]$ and $x[1,l-1]$ are uniformly distributed.
\label{lower_bound_ac}
\end{lemma}

\begin{lemma}
Let $P$ be a randomized  protocol for $\MHAI{n}{1}$ 
with transcript $\Pi$ satisfying
$\mi(D:\Pi_B|X,K) \leq \alpha$.
Then
$$\Exp_{x[1,l-1],l} \hel^2(\Pi(x[1,l-1]0X[l+1,n],l,0),  \Pi(x[1,l-1]0X[l+1,n],l,1))
\leq 12\alpha,$$
where $l\in[\frac{n}{2}+1,n]$ and $x[1,l-1]$ are uniformly distributed.
\label{lower_bound_b}
\end{lemma}

We now end with the main lemma which combines both previous ones and
applies the cut-and-paste property, where Players $A,C$ are grouped.
\begin{lemma}
Let $P$ be a randomized  protocol for $\MHAI{n}{1}$ with bounded
error $\epsilon$, and transcript $\Pi$ satisfying
$ |\Pi| \leq \alpha n  $ and
$\max \left\{ I(D:\Pi_B|X,K),
I(K,D:\Pi_C|X) \right\} \leq \alpha$.
Then $\alpha\geq (1-2\eps)/10$.
\label{lower_bound_epsilon}
\end{lemma}
\begin{proof}
Let $L$ be a uniform integer random variable in $[\frac{n}{2}+1,n]$.
Remind that we enforce the output of $P$ to be part of $\Pi$.
Therefore, any player, and in particular $B$, can compute $f$ with bounded error $\eps$ given $\Pi$.
Since $f(x[1,l-1]0X[l+1,n],l,0)=0$ and $f(x[1,l-1]1X[l+1,n],l,1)=1$, the error parameter $\eps$ must satisfies
$$\Exp_{x[1,l-1],l} \lVert \Pi(x[1,l-1]0X[l+1,n],l,0) -  \Pi(x[1,l-1]1X[l+1,n],l,0) \rVert_1 \geq 2(1-2\eps).$$
The rest of the proof consists in upper bounding the LHS by  $19\alpha$.

Applying the triangle inequality and that $(u+v)^2\leq 2(u^2+v^2)$ on
the inequalities of Lemmas~\ref{lower_bound_ac} and~\ref{lower_bound_b} gives
$$\Exp_{x[1,l-1],l} \hel^2(\Pi(x[1,l-1]0X[l+1,n],l,0),  \Pi(x[1,l-1]1X[l+1,n],l,1))
\leq 30\alpha.$$
We then apply the cut-and-paste property by considering $(A,C)$ as a single player with transcript $\Pi_{A,C}$.
Therefore
$$\Exp_{x[1,l-1],l} \hel^2(\Pi(x[1,l-1]0X[l+1,n],l,1),  \Pi(x[1,l-1]1X[l+1,n],l,0))
\leq 30\alpha.$$
Combining again with the inequality from Lemma~\ref{lower_bound_b} gives
$$\Exp_{x[1,l-1],l} \hel^2(\Pi(x[1,l-1]0X[l+1,n],l,0),  \Pi(x[1,l-1]1X[l+1,n],l,0))
\leq 42\alpha.$$
Last, we get the requested upper bound by
using the connexion between the Hellinger distance and the
$\ell_1$-distance, and the convexity of the square function.
\end{proof}

\section{Bidirectional streaming algorithm for $\PQ$}\label{sec:algo}
Remember that in this section our stream is given without any timestamps.
Therefore we consider in this section only streams $w$ of $\ins(a),\ext(a)$, where $a\in [0,U]$.
For the sake of clarity, we assume for now that the stream has no duplicate.
Our algorithms can be extended to the general case, but the technical difficulties
shadow the main ideas.

Up to padding we can assume  that $N$ is a power of $2$: we append a sequence of
$\ins(a)\ext(a)\ins(a+1)\ext(a+1)\ldots$ of suitable length, 
where $a$ is large enough so that there is no duplicate
(assuming that $w$ is of even size, otherwise $w\not\in\PQs(U)$).
We use $\Order(\log N)$ bits of memory to store, after the first pass, the number of letters padded. 

We use a hash function based on the one used by the Karp-Rabin algorithm for pattern matching.
For all this section, let $p$ be a prime number in
$\{\max(2U+1,N^{c+1}),\dots,2 \max(2U+1,N^{c+1})\}$, for some fixed constant $c\geq 1$.
Since our hash function is linear we only define it for single insertion/extraction as
\[
\hash(\ins(a))=\alpha^a \mod p, \quad\text{and}\quad \hash(\ext(a))=-\alpha^a \mod p,
\]
where  $\alpha$ is a randomly chosen integer in $[0,p-1]$. 
This is the unique source of randomness of our algorithm.
A hashcode $h$ {\em encodes} a sequence $w$ if $h=\hash(w)$
as a formal polynomial in $\alpha$. In that case we say that $h$ {\em includes} $w[i]$, for all $i$.
Moreover $w$ is {\em balanced} if the same integers have been inserted and extracted.
In that case it must be that $h=0$. 
We also say that $h$ is balanced it it encodes a balanced sequence $w$.
The converse is also true with high probability
by the Schwartz-Zippel lemma.
\begin{fact}\label{fact:schwartz}
Let $w$ be some unbalanced sequence.
Then $\Pr(\hash(w)=0)\leq \frac{N}{p} \leq \frac{1}{N^c}$.
\end{fact}

The forward-pass algorithm was introduced in~\cite{cckm10}, but the reverse-pass one is even simpler.
As a warming up, we start by introducing the later algorithm. In order to keep it simple to understand, we do not
optimize it fully. 
Last  define the instruction $\update(h,v)$ that returns $(h+\hash(v)\mod p)$ {\em and} updates $h$ to that value.

\subsection{One-reverse-pass algorithm for $\PQ$}

Our algorithm decomposes the stream $w$ into blocks. We call
a valley an extraction $w[t]=\ext(a)$ with $w[t+1]=\ins(b)$. A new
block starts at every valley. To the $i$-th block we associate a
hashcode $h_i$ and an integer $m_i$. Hashcode $h_i$ encodes all the extractions within the block and the
matching insertions. Integer $m_i$ is the minimum of extractions
in the block. With the values $(m_i)_i$, one can encode insertions
in the correct $h_i$ if $w \in \PQs$. Observe that we use index notations for block
indices and bracket notations for stream positions.

Algorithm \ref{AlgoReverse} uses memory space $\Order(r)$,
where $r$ is the number of valleys in $w$. We could make it run with
memory space $\Order(\sqrt{N\log N})$ by reducing the number of valleys as
in \cite{cckm10}. We do not need to as we use another compression in the
two-pass algorithm.

\begin{lstlisting}[caption={One-reverse-pass algorithm for
      $\PQ$},label=AlgoReverse,captionpos=t,float,abovecaptionskip=-\medskipamount,mathescape]
$m_0 \gets - \infty$; $h_0 \gets 0$; $t \gets N$; $i \gets 0$  // $i$ is called the block index
While $t > 0$
    If $w[t]=\ins(a)$    
       $k \gets \max\{j \leq i : m_j \leq a\}$;(*@\label{Set_k}@*)    //Compute the hashcode index of $a$
       $\update(h_k,w[t])$
    Else $w[t]=\ext(a)$
       If $w[t+1]=\ins(b)$    //This is a valley. We start a new block
          $i \gets i+1$; $m_i \gets a$; $h_i \gets 0$  //Create a new hashcode
       Else $w[t+1]=\ext(b)$
          Check($a \geq b$) (*@\label{Check1}@*)   //Check that extractions are well-ordered
       $\update(h_i,w[t])$  
    $t \gets t-1$
For $j=0$ to $i$:  Check($h_j=0$) (*@\label{Check2}@*)//Check that hashcodes are balanced w.h.p.
Accept   // $w$ succeeded to all checks
\end{lstlisting}

We first state a crucial property of Algorithm~\ref{AlgoReverse}, and
then show that it satisfies Theorem~\ref{theorem_AlgoReverse}, when
there is no duplicate. We remind that we process the stream from right to left.

\begin{lemma}
\label{lemma_matching_pair}
Consider Algorithm~\ref{AlgoReverse} right after processing $\ins(a)$.
Assume that $\ext(a)$ has been already processed.
Let $h_k,h_{k'}$ be the respective hashcodes including $\ext(a),\ins(a)$.
Then $k=k'$ if and only if all $\ext(b)$ occurring between $\ext(a)$ and $\ins(a)$ satisfy $b>a$.
\end{lemma}

\begin{theorem}
\label{theorem_AlgoReverse}
There is a  $1$-reverse-pass randomized streaming algorithm for
$\PQs(U)$ with memory space $\Order(r(\log N + \log U))$ and one-sided bounded error $N^{-c}$,
for inputs  of length $N$ with $r$ valleys, and any constant $c>0$.
\end{theorem}

\begin{proof}
We show that Algorithm~\ref{AlgoReverse} suits the conditions, assuming there is no duplicate.
Let $w \in \PQs(U)$. Then $w$ always passes the test at line~\ref{Check1}. 
Moreover, by Lemma~\ref{lemma_matching_pair}, each insertion
$\ins(a)$ is necessarily in the same hashcode than its matching extraction $\ext(a)$. 
Therefore, all hashcodes equal~$0$ at line~\ref{Check2} since they are balanced.
In conclusion, the algorithm accepts $w$ with probability $1$.

Assume now that $w \not \in \PQs$. First we show that unbalanced $w$
are rejected with high probability, that is at least $1-N^{-c}$, 
at line~\ref{Check2}, if they are not rejected before. Indeed, since
each $w[t]$ is encoded in some $h_j$, 
at least one $h_j$ must be unbalanced. Then by
Fact~\ref{fact:schwartz}, the algorithm rejects w.h.p.
We end the proof assuming $w$ balanced.
We remind that we process the stream from right to left.
The two remaining possible errors are: 
(1) $\ins(a)$ is processed before $\ext(a)$, for some $a$; 
and (2) $\ext(a),\ext(b),\ins(a)$ are processed in this order with
$b<a$ and possibly intermediate insertions/extractions.
In both cases, we show that some hashcodes are unbalanced at
line~\ref{Check2}, and therefore fail the test w.h.p by
Fact~\ref{fact:schwartz}, except if the algorithm rejects before.

Consider case (1). Since $\ins(a)$ is processed before $\ext(a)$, 
there is at least one valley between $\ins(a)$ and
$\ext(a)$. Therefore $\ins(a)$ and $\ext(a)$ are encoded into two
different hashcodes, that are unbalanced at line~\ref{Check2}.

Consider now case (2). Lemma~\ref{lemma_matching_pair} gives that
$\ext(a)$ and $\ins(a)$ are encoded in two different hashcodes,
that are again unbalanced at line~\ref{Check2}.
\end{proof}

\subsection{Bidirectional two-pass algorithm}
Algorithm~\ref{Algo2Pass} performs one pass in each direction using Algorithm~\ref{1Pass}.
We use the hierarchical data structure of~\cite{mmn10} in order to reduce the number of blocks.
A block of size $2^i$ is of the form $[ (q-1)2^i +1,q2^i]$, for $1\leq q\leq N/2^i$.
Observe that, given two such blocks, either they are disjoint or one is included in the other.
We decompose dynamically the letters of $w$, that have been already processed, 
into nested blocks of $2^i$ letters as follows.
Each new processed letter of $w$ defines a new block. When two blocks have same size, they merge. 
All processed blocks are pushed on a stack. Therefore, only the two topmost blocks of the stack may potentially merge.
Because the size of each block is a power of $2$ and at most two blocks have the same size (before merging),
there are at most $\log N +1$ blocks at any time.

Moreover, since our stream size is a power of $2$, all
blocks eventually appear in the hierarchical decomposition, whether we
read the stream from left to right or from right to left. 
In fact, if two same-sized blocks appear simultaneously in one decomposition before merging, the same is true in the
other decomposition. This point is crucial for our analysis. 

Algorithm~\ref{1Pass} uses the following description of a block $B$:
its hashcode $h_B$ , the minimum $m_B$ of its extractions, and  its
size $\ell_B$. For the analysis, we also note $t_B$ the index
such that $w[t_B]=m_B$. Among those parameters, only $h_B$ can change 
without $B$ being merged with another block. 
On the pass from right to left, all extractions from the block 
and the matching insertions are included in $h_B$. 
On the pass from left to right, insertions are included in the
hashcode of the earliest possible block where they could have been,
and the extractions are included with their matching insertions. The
minimums $(m_B)_B$ are used to decide where to include values (except
extractions on the pass from right to left). Observe that it is
important to check that $h_B=0$ whenever possible and not at the end
of the execution of the algorithm, since only one block is left at the
end.

When there is some ambiguity, we denote by $h_B^\rightarrow$ and $h_B^\leftarrow$ the hashcodes
for the left-to-right and right-to-left passes. Observe that $m_B,t_B,\ell_B$ are identical in both directions.

\begin{lstlisting}[caption={Bidirectional $2$-pass algorithm for $\PQ$},label=Algo2Pass,captionpos=t,float,abovecaptionskip=-\medskipamount,mathescape]
OnePassAlgorithm($w$)   reading stream from left to right 
OnePassAlgorithm($w$)   reading stream from right to left 
Accept   // $w$ succeeded to all checks
\end{lstlisting}

\begin{lstlisting}[caption={OnePassAlgorithm},label=1Pass,captionpos=t,float,abovecaptionskip=-\medskipamount,mathescape]
$S \gets []$; 
If left-to-right-pass Then Push($S$,$(0,-\infty,0)$)  // Initialization of $S$
While stream is not empty
    Read(next letter $v$ on stream)  // See below
    While the $2$ topmost elements of $S$ have same block size $\ell$
       $(h_1,m_1,\ell) \gets$Pop($S$); $(h_2,m_2,\ell) \gets$Pop($S$)
       Push($S$,$(h_1+h_2  \mod p,\min(m_1,m_2),2\ell)$)  // Merge of $2$ blocks
If left-to-right-pass Then Check($S = [(0,-\infty,0),(0,0,N)]$)
Else Check($S = [(0,0,N)]$)}
Return

Function Read(v):
Case $v=\ins(a)$  // When reading an insertion
    Let $(h,m,\ell)$ be the first item of $S$ from top such that $a\geq m$ (*@\label{BlockIns}@*)
    Replace $(h,m,\ell)$ by $(\update(h,v),m,\ell)$
    Push $($S$,(0,+\infty,1)$)
Case $v=\ext(a)$ and left-to-right-pass // When reading an extraction
    For all items $(h,m,\ell)$ on $S$ such that $m>a$: Check($h=0$) (*@\label{CheckLR}@*)
    Let $(h,m,\ell)$ be the first item of $S$ from top such that $a> m$
    Replace $(h,m,\ell)$ by $(\update(h,v),m,\ell)$
    Push($S$,$(0,a,1)$)
Case $v=\ext(a)$ and right-to-left-pass // When reading an extraction
    For all items $(h,m,\ell)$ on $S$ such that $m>a$: Check($h=0$) (*@\label{CheckRL}@*)
    Push($S$,$(\hash(v),a,1)$)
\end{lstlisting}

\begin{proof}[Proof of Theorem~\ref{main2}]
We show that Algorithm~\ref{Algo2Pass} suits the conditions, assuming there is no duplicate.
The space constraints are satisfied because each element of $S$
takes space $\Order(\log N + \log U)$ and $S$ has size at most 
$\log N+1$. The processing time is from inspection.

As with Theorem~\ref{theorem_AlgoReverse}, inputs in $\PQs(U)$
are accepted with probability $1$, and unbalanced inputs
are rejected with high probability (at least $1 - N^{-c})$. Let $w \not \in \PQs$ be
balanced. 
For ease of notations, let $w[-1]=\ins(-\infty)$ and $w[0]=\ext(-\infty)$.
Then, there are $\tau < \rho$ such that $w[\tau]=\ext(b)$, $w[\rho]=\ext(a)$,
$a>b$, and $w[t] \not =\ins(a)$ for all $\tau < t < \rho$. 

Among those pairs $(\tau,\rho)$,  consider the ones with the
smallest $\rho$. From those,  select the one with the smallest
$b$, with $w[\tau]=\ext(b)$.
Let $B$, $C$ be the largest possible disjoint blocks such
that $\tau$ is in $B$ and $\rho$ in $C$. Then $B$ and $C$ have
same size, are contiguous, and appear simultaneously in each direction before they merge.
Let $\rho'$ and $\tau'$ be such that $w[\rho']=\ins(a)$ and
$w[\tau']=\ins(b)$. 
The minimality of $\rho$ and the minimality of $b$
guarantee that $w[t]$ is an insertion for all $\tau < t < \rho$.
Indeed if $w[t]=\ext(c)$ either $b> c$, which contradicts the minimality of $b$, or $c > b$ and
$(\tau,t)$ contradicts the minimality of $\rho$. 
In particular, $t_C
\geq \rho$ and $t_B \leq \tau$. Similarly $\tau < \tau'$,
otherwise $\tau$ would be a better candidate than $\rho$.

We distinguish three cases based on the position $\rho'$ of
$\ins(a)$ (see Figure~\ref{fig:main2}): $\rho' \not \in [t_B,t_C]$, $t_B < \rho' < \tau$,
and $\rho < \rho' < t_C$. These cases determine in which hashcode
$\ins(a)$ is included. We analyze Algorithm~\ref{1Pass} when some letter is processed
before blocks potentially merge.

{\em Case~1}: $\rho' \not \in [t_B,t_C]$. One can prove that $h_B^\rightarrow$ is unbalanced when
$w[t_C]$ is processed and that $h_C^\leftarrow$ is unbalanced when
$w[t_B]$ is processed; therefore Algorithm~\ref{1Pass}  detects w.h.p.
$h_B^\rightarrow \not = 0$ or $h_C^\leftarrow \not = 0$ depending on whether $m_B>m_C$
(see Lemma~\ref{lemme a} in Appendix~\ref{app:algo}).

{\em Case~2}: $t_B < \rho' < \tau$. 
We show that when Algorithm~\ref{1Pass} processes $w[t_B]=\ext(m_B)$,
it checks $h_{D}^\leftarrow=0$ at line~\ref{CheckRL} for some $h_{D}^\leftarrow$ 
including $\ins(a)$ but not $\ext(a)$. Thus it rejects w.h.p.

When $w[\rho']=\ins(a)$ is processed on the right-to-left pass, 
$\tau\in B_1$ with $B_1$ a block in the stack. $\tau \in B$, therefore
$B_1$ intersects $B$. Because $B_1 \not \subseteq B$, we have $B_1 \subseteq B$.
Because $w[\tau]=\ext(b)$, we have $a>b\geq m_{B_1}$, and block~$B_1$
is eligible at line~\ref{BlockIns} of Algorithm~\ref{1Pass},
meaning that $w[\rho']=\ins(a)$ is included in either $h_{B_1}^\leftarrow$  
or a more recent hashcode $h_{B_2}^\leftarrow$.
Since $\rho' \in B$, again $B_2\subseteq B$.
Last, when Algorithm~\ref{1Pass} processes $w[t_B]=\ext(m_B)$, since
we are still within $B$, some hashcode $h_{B_3}$, with $B_3 \subseteq
B$, includes $w[\rho']$. 
Moreover, $h_{B_3}^\leftarrow$ does not include $w[\rho]=\ext(a)$
since $\rho\in C$ and $C$ comes before $B$. Last, $m_{B_3} > m_B$, by
definition of $m_B$.
Hence, Algorithm~\ref{1Pass} checks
$h_{B_3}^\leftarrow=0$ at line~\ref{CheckRL} when processing
$w[t_B]$. $B_3$ satisfies the conditions for $D$ when $w[t_B]$ is
processed, and Algorithm~\ref{1Pass} rejects w.h.p.

{\em Case~3}: $\rho < \rho'< t_C$. The proof is the same as case $2$,
replacing $\tau$, $B$, $B_1$, $B_2$, $B_3$, $h_{B_1}^\leftarrow$,
$h_{B_2}^\leftarrow$, $h_{B_3}^\leftarrow$, $t_B$, $C$ with $\rho$,
$C$, $C_1$, $C_2$, $C_3$, $h_{C_1}^\rightarrow$, $h_{C_2}^\rightarrow$,
$h_{C_3}^\rightarrow$, $t_C$, $B$ and line~\ref{CheckRL} with line~\ref{CheckLR}. Note that we only have $a \geq
m_{C_1}$ this time, so it is important that the inequality at
line~\ref{BlockIns} is large and not strict.
\end{proof}

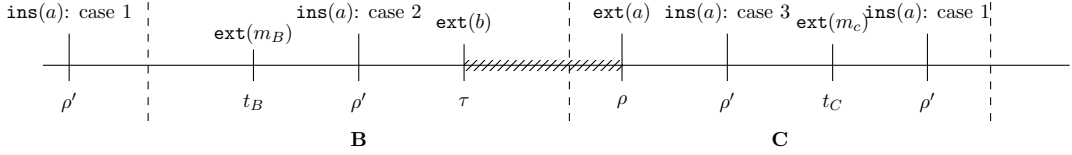
\begin{figure}
\hspace*{1cm}\begin{tikzpicture}[scale=0.7,every node/.style={scale=0.7}]
\draw (0,0) -- (19.5,0);
\draw[dashed] (10,1.2) -- (10,-1.2);
\draw[dashed] (2,1.2) -- (2,-1.2);
\draw[dashed] (18,1.2) -- (18,-1.2);
\draw (8,0.4) -- (8,-0.3);
\draw (11,0.6) -- (11,-0.3);
\draw (4,0.3) -- (4,-0.3);
\draw (15,0.4) -- (15,-0.3);
\draw (0.5,-0.3) -- (0.5,0.6);
\draw (6,-0.3) -- (6,0.6);
\draw (13,-0.3) -- (13,0.6);
\draw (16.8,-0.3) -- (16.8,0.6);

\node at (8,0.8) {$\ext(b)$};
\node at (11,1) {$\ext(a)$};
\node at (4,0.6) {$\ext(m_B)$};
\node at (15,0.8) {$\ext(m_c)$};
\node at (8,-0.7) {$\tau$};
\node at (11,-0.7) {$\rho$};
\node at (4,-0.7) {$t_B$};
\node at (15,-0.7) {$t_C$};
\node at (0.5,1) {$\ins(a)$: case 1};
\node at (6,1) {$\ins(a)$: case 2};
\node at (13,1) {$\ins(a)$: case 3};
\node at (16.8,1) {$\ins(a)$: case 1};
\node at (0.5,-0.7) {$\rho'$};
\node at (6,-0.7) {$\rho'$};
\node at (13,-0.7) {$\rho'$};
\node at (16.8,-0.7) {$\rho'$};
\node at (6,-1.4) {$\mathbf{B}$};
\node at (14,-1.4) {$\mathbf{C}$};

\fill[pattern=north east lines] (8,0.1) -- (11,0.1) -- (11,-0.1) --
(8, -0.1);

\end{tikzpicture}
\caption{\label{fig:main2}Relative positions of insertions and
  extractions used in the proof of Theorem~\ref{main2}}
\end{figure}

\subsection{Generalization when duplicates occur}

We maintain two additional parameters $\delta_B$ and $C_B$ for each block
$B$. The difference between the number of insertions and
extractions included in $h_B$ is stored in $\delta_B$. Whenever $\delta_B=0$, we check
$h_B=0$. The number of unmatched occurrences of $\ins(m_B)$ for the left-to-right pass
(resp. $\ext(m_B)$ for the right-to-left pass) is stored in $C_B$. 
We can then appropriately
determine whether each $\ext(m_B)$ (resp. $\ins(m_B)$) should be included in $h_B$.

The change on the criterion of line~\ref{BlockIns} of
Algorithm~\ref{1Pass} makes the proof of case 3 of the theorem longer
and breaks the symmetry. 

\section*{Acknowledgements}
Authors would like to thank Rahul Jain and Ashwin Nayak for sharing their intuition and possible extensions of~\cite{jn10}.
In particular, N.F. thanks Ashwin Nayak for having hosted him at IQC, University of Waterloo.
They also thanks Christian Konrad, Wei Yu, Qin Zhang for related discussions,
and Andrew McGregor for motivating us to study restricted instances of $\PQ$.

\bibliography{fm12}

\appendix

\section{Missing proofs for the lower bound}\label{app:lb}
We start by proving the lemma relating  the streaming complexity of
deciding if an instance of $\RD{m}{n}$ belongs to $\PQT(3mn)$ to the
communication complexity of $\MHAI{n}{m}$.
\begin{proof}[Proof of Lemma~\ref{stream-comm}]
Assume that there exists a $p$-pass randomized streaming algorithm with memory space $s(m,n)$,
that decides if an instance of $\RD{m}{n}$ belongs or not to $\PQT(3nm)$.
Each instance of $\RD{m}{n}$ can be encoded by an input of $\MHAI{n}{m}$,
where each of the $3m$ players has one part of it.
Then, the rest of the proof consists in showing how the players can use the algorithm in order to
construct a protocol that satisfies the required properties of the lemma.

Each player simulates alternatively the algorithm.
A player performs the simulation until the algorithm reaches the part of the input of the next player.
Then the player sends the current state of the algorithm, so that the next player can continue the simulation. 
Since the algorithm uses at most memory space $s(m,n)$, the current state can be encoded using $s(m,n)$ bits.
Each pass corresponds to one round of communication, implying the result.
\end{proof}

Before giving the next missing proofs of Section~\ref{sec:lb}, we state some useful properties
of entropy and mutual information that we need. See~\cite{jn10} for more information.
\begin{fact}\label{fact:mi}
Let $X, Y, Z, R$ be random variables such $X$ and $Z$ are independent
when conditioning on $R$, namely when conditioning on $R=r$, for each
possible values of $r$.
Then $\mi(X : Y | Z, R) \geq \mi(X : Y | R)$.
\end{fact}
\begin{proof}
From the definition of mutual information and the independence of $X,Z$ when conditioning on $R$,
we get that
$$\mi( X : Y | Z,R ) = \h(X | Z,R) - \h(X|Y,Z,R) = \h(X |R) - \h(X|Y,Z,R).$$
Using that entropy can only decrease under conditioning, and using
again the definition of mutual information, we conclude by bounding the last term as
$$\h(X | R) - \h(X|Y,Z,R) \geq \h(X | R) - \h(X|Y,R) = \mi( X : Y | R ).$$
\end{proof}

\begin{proposition}[Chain rule]
Let $X,Y,Z,R$ be random variables.
Then $\mi(X,Y:Z | R) = \mi(X:Z|R) + \mi(Y:Z|X,R)$.
\end{proposition}

\begin{proposition}[Data processing inequality]
Let $X,Y,Z,R$ be random variables such that $R$ is independent from $X,Y,Z$.
Then $\mi(X:Y | Z) \geq \mi(f(X,R):Y | Z)$, for every function $f$.
\end{proposition}
Note that the previous property is usually stated with no variable $T$. Nonetheless, since $T$ is independent
from the other variables, we have $\mi(X:Y | Z)=\mi(X,R:Y | Z)$, and
then we can apply the usual data processing inequality.

We can now prove our three lemmas.

\begin{proof}[End of proof of Lemma~\ref{multi_dir_sum}]
Let $\Pi,\Pi'$ be the respective transcripts of $P,P'$.
For convenience, note $\Pi_{C_{m+1}}=\Pi_{B_m}$, $\Pi_{B_{0}}=\Pi_{C_m}$ and
$\Pi_{C_{m+1}}=\Pi_{A_1}$. 
Remind that the public coins of a protocol are included in its transcript.

First, each player of $P'$ sends $3$ messages by round, and there are $(p+1)$ rounds.
Since each message has size at most $s(m,n)$, we derive that the length of $\Pi'$ is at most
$3(p+1)s(m,n)$.

Then, in order to prove that there is only a small amount of information in the transcripts of Bob and Charlie,
we show a direct sum of some appropriated notion of information cost. 
Consider first the transcript of Player $C_1$. Because of the restriction on the size of his messages, 
we know that $|\Pi_{C_1}|\leq (p+1) s(m,n)$. From this we derive a first inequality
on the amount of information this transcript can carry,
using that the entropy of a variable is at most its bit-size:
$$ \mi(\mathbf{K},\mathbf{D}:\Pi_{C_1}|\mathbf{X}) \leq |\Pi_{C_1}| \leq (p+1) s(m,n).$$
We now use the chain rule in order to get a bound about the information carried by $P'$ on a single instance.
\begin{eqnarray*}
\mi(\mathbf{K},\mathbf{D}:\Pi_{C_1}|\mathbf{X})
& = & \sum_{j=1}^m
\mi((K_i,D_i)_{j\geq i}:\Pi_{C_1}|\mathbf{X},(K_i,D_i)_{i<j})
\quad\text{(by chain rule)}
\\
& \geq &
\sum_{j=1}^m \mi(K_j,D_j:\Pi_{B_{j-1}}|\mathbf{X},(K_i,D_i)_{i<j})
\quad\text{(by data processing inequality)}
\\
& \geq & \sum_{j=1}^m \mi(K_j,D_j:\Pi_{B_{j-1}}|\mathbf{X}) 
\quad\text{(by Fact~\ref{fact:mi})}
\\
& = & m\times\mi(K_J,D_J:\Pi_{B_{J-1}}|\mathbf{X},J) 
\quad\text{(by conditioning on $J$)}
\\
& = & m\times\mi(K_J,D_J:\Pi_{B_{J-1}} ,J,(X_i)_{i\neq J} | X_J) 
\ \text{(independence of $J,(X_i)_{i\neq J}$)}
\\
&=& m\times \mi(K,D:\Pi'_C|X)
\quad\text{(since $J,(X_i)_{i\neq J}$ are public coins of $P'$).}
\end{eqnarray*}

We then do similarly for Player $B_m$ and therefore conclude the proof.
First the size bound on messages of $B_m$ gives 
$\mi(\Pi_{B_m}:\mathbf{D}|\mathbf{X},\mathbf{K})\leq (p+1) s(m,n)$.
Then as before we get:
\begin{eqnarray*}
\mi(\mathbf{D}:\Pi_{B_m}|\mathbf{X},\mathbf{K})
& = &
\sum_{j=1}^m \mi(D_j:\Pi_{B_m}|\mathbf{X},\mathbf{K},(D_i)_{i> j})
\geq
\sum_{j=1}^m \mi(D_j:\Pi_{C_{j+1}}|\mathbf{X},\mathbf{K},(D_i)_{i>j})\\
& \geq &
m\times \mi(D_J:\Pi_{C_{J+1}},J,(X_i)_{i\neq J}|X_J,K_J)
 = m \times \mi( D:\Pi'_B|X,K).
\end{eqnarray*}
\end{proof}

\begin{proof}[Proof of Lemma~\ref{lower_bound_ac}]
From the second hypothesis and the data processing inequality we get that
$\mi(K,D:\Pi_{A,C}|X) \leq \alpha$, which after applying the average encoding leads to
$\Exp_{x,k,d} \hel^2(\Pi_{A,C}(x,k,d),\Pi_{A,C}(x,K,D))\leq \kappa \alpha$.
We now restrict $\mu_0$ by conditioning on $D=1$. Then $(X,K)$ is uniformly distributed.
Moreover, since $D=1$ with probability $3/4$ on $\mu_0$, we get
$\Exp_{x,k } \hel^2(\Pi_{A,C}(x,k,1),\Pi_{A,C}(x,K,1)) \leq \frac{4}{3}\kappa \alpha$.
Let $J,L$ be uniform integer random variables respectively in $[2,\frac{n}{2}]$ and $[\frac{n}{2}+1,n]$.
Then the above implies
$\Exp_{x,j } \hel^2(\Pi_{A,C}(x,j,1),\Pi_{A,C}(x,K,1)) \leq \frac{8}{3}\kappa \alpha$
and
$\Exp_{x,l } \hel^2(\Pi_{A,C}(x,l,1),\Pi_{A,C}(x,K,1)) \leq \frac{8}{3}\kappa \alpha$.
Applying the triangle inequality and that $(u+v)^2\leq 2(u^2+v^2)$, we get
$$\Exp_{x,j,l} \hel^2(\Pi_{A,C}(x,j,1),\Pi_{A,C}(x,l,1)) \leq \tfrac{32}{3}\kappa \alpha.$$
Using the convexity of $\hel^2$, we finally obtain for $b=0,1$:
$$\Exp_{x[1,l-1],j,l}
\hel^2(\Pi_{A,C}(x[1,l-1]bX[l+1,n],j,1),\Pi_{A,C}(x[1,l-1]bX[l+1,n],l,1))
\leq \tfrac{64}{3}\kappa \alpha.$$

Now the chain rule allow us to measure the information about a single bit in $\Pi_{A,C}$ as
\begin{eqnarray*}
\mi(X[L]:\Pi_{A,C}(X,J,1)|X[1,L-1])&=&\Exp_{l\gets L}\mi(X[l]:\Pi_{A,C}(X,J,1)|X[1,l-1])\\
&=&\frac{2}{n}\times \mi(X[\tfrac{n}{2}+1,n]:\Pi_{A,C}(X,J,1) | X[1,\tfrac{n}{2}]).
\end{eqnarray*}
Since the entropy of a variable is at most its bit-size, we get that the last term is upper bounded by
$|\Pi_{A,C}|$, which is at most $\alpha n$ by the first hypothesis.
Then as before, the average encoding and the triangle inequality lead to
$$\Exp_{x[1,l-1],j,l} \hel^2(\Pi_{A,C}(x[1,l-1]0X[l+1,n],j,1),\Pi_{A,C}(x[1,l-1]1X[l+1,n],j,1))
\leq 16\kappa \alpha.$$

Combining gives
$$\Exp_{x[1,l-1],l} \hel^2(\Pi_{A,C}(x[1,l-1]0X[l+1,n],l,1),  \Pi_{A,C}(x[1,l-1]1X[l+1,n],l,1))\leq 28\alpha.$$
Let $R_B$ be the random coins of $B$. Since they are independent from all variables, including the messages, 
the previous inequality is still true when we concatenate $R_B$ to $\Pi_{A,C}$. Then $\Pi_B$ is uniquely
determined from $R_B$ once $K,D,X[1,K-1]$ are fixed, which is the case in that inequality.
Therefore replacing $R_B$ by $\Pi_B$ can only decrease the distance, concluding the proof.
\end{proof}

\begin{proof}[Proof of Lemma~\ref{lower_bound_b}]
Using the data processing inequality and the hypothesis we get that
$\mi(D:\Pi|X,K))\leq\alpha$. Therefore by average encoding,
$\Exp_{x,k,d} \hel^2(\Pi(x,k,d),\Pi(x,k,D))\leq \kappa \alpha$.

Let $L$ be a uniform integer random variable in $[\frac{n}{2}+1,n]$.
Then $\Exp_{x,l,d} \hel^2(\Pi(x,l,d),\Pi(x,l,D))\leq 2\kappa \alpha$.
Using the convexity of $\hel^2$ and the fact that $X[l]$ is a uniform random bit, we derive
$$\Exp_{x[1,l-1],l,d} \hel^2(\Pi(x[1,l-1]0X[l+1,n],l,d),\Pi(x[1,l-1]0X[l+1,n],l,D))\leq 4\kappa \alpha.$$

Since $D=0$ with probability $1/2$ when $X[l]=0$ and $K=l$, we finally get the two inequalities
$$\Exp_{x[1,l-1],l} \hel^2(\Pi(x[1,l-1]0X[l+1,n],l,0),\Pi(x[1,l-1]0X[l+1,n],l,D))\leq 8\kappa \alpha,$$
$$\Exp_{x[1,l-1],l} \hel^2(\Pi(x[1,l-1]0X[l+1,n],l,1),\Pi(x[1,l-1]0X[l+1,n],l,D))\leq 8\kappa \alpha,$$
leading to the conclusion using the triangle inequality and that $(u+v)^2\leq 2(u^2+v^2)$.
\end{proof}

\section{Missing proofs for the algorithm}\label{app:algo}
We start by proving the property of Algorithm~\ref{AlgoReverse} we
use in the proof of Theorem~\ref{theorem_AlgoReverse}.
\begin{proof}[Proof of Lemma~\ref{lemma_matching_pair}]
Remind again, that we process the stream from right to left in this proof, and
that $h_k,h_{k'}$ are the respective hashcodes including $\ext(a),\ins(a)$.
First assume that all $\ext(b)$ between $\ext(a)$ and $\ins(a)$ satisfy $b>a$.
Let $i$ be the current block index while processing $\ins(a)$. 
Observe that $k$ is the current block index right after processing $\ext(a)$. 
Since $\ext(a)$ is processed before $\ins(a)$ and since there is a
valley between $\ext(a)$ and $\ins(a)$, we have $k< i$.

We prove that $k' = \max \{j \leq i | m_j \leq a\}=k$. The first
equality is from line \ref{Set_k} of Algorithm~\ref{AlgoReverse}. We now prove the 
second equality. 
For each $j \in \{k+1,\dots,i\}$, value $m_j$ is extracted between $\ext(a)$ and $\ins(a)$.
Then, our assumption leads to $m_j > a$. 
Moreover, because the algorithm checks at line~\ref{Check1} that
extraction sequences included in the same hashcode are decreasing,
we have $m_k \leq a$, leading to the second equality.

We now prove the converse by contrapositive. 
Assume that some $\ext(b)$ between $\ext(a)$ and $\ins(a)$ satisfies
$b\leq a$. Since we forbid duplicates, in fact $b<a$.
Let $j$ be the current block index right after processing
$\ext(b)$. Then line~\ref{Check1} ensures that $m_j \leq b$. 
Again, $k$ is the current block index right after processing
$\ext(a)$, and therefore $k\leq j$.
If $k=j$, then the extraction sequence is not decreasing and
line~\ref{Check1} rejects, contradicting the hypotheses that the
algorithm has not rejected yet after processing $\ins(a)$. Therefore
$k<j$. But, line~\ref{Set_k} and the fact that $m_j \leq b$ imply that
$k'\geq j$, and therefore $k<k'$.
\end{proof}

We now give the missing part of the proof of Theorem~\ref{main2}.

\begin{lemma}
\label{lemme a}
If $\rho' \not \in [t_B,t_C]$, then
Algorithm~\ref{Algo2Pass} rejects $w$ with probability at 
least $1-N^{-c}$. 
\end{lemma}

\begin{proof} 
We prove that $h_B^\rightarrow$ is unbalanced when
$w[t_C]$ is processed and that $h_C^\leftarrow$ is unbalanced when
$w[t_B]$ is processed. From that, we deduce that the algorithm rejects
with high probability unless $m_B \leq m_C$ and $m_C \leq m_B$,
i.e. $m_B=m_C$, which is impossible because $w$ has no duplicates and
$B$ and $C$ are disjoint.

Indeed if $m_C < m_B$ then Algorithm~\ref{1Pass} checks that 
$h_B^\rightarrow=0$ at line~\ref{CheckLR} when processing $w[t_C]$,
and rejects with high probability because $h_B^\rightarrow$ is
unbalanced. Similarly, if $m_C < m_B$, it rejects with high
probability at line~\ref{CheckRL} when processing $w[t_B]$ on the
right-to-left pass.

Now we only have to prove that $h_B^\rightarrow$ (resp. $h_C^\leftarrow$) is unbalanced when
$w[t_C]$ (resp. $w[t_B]$) is processed.
Let us assume there exists $B_1 \subsetneq B$ such that $\ins(a)$ is
included in $h_{B_1}^\rightarrow$ when $w[t_B]$ is processed. Then, by
definition of $m_B$, $m_{B_1} > m_B$. Moreover, $\rho \in C$, so $w[\rho]=\ext(a)$ is
not processed yet and not included in $B_1$. Therefore,
Algorithm~\ref{1Pass} checks $h_{B_1}^\rightarrow=0$ at
line~\ref{CheckLR}, and
rejects w.h.p. We can now
assume that there is no such 
$B_1 \subsetneq B$, and therefore that
$h_B$, does not include $\ins(a)$ when $w[t_C]$ is
processed. Since $h_{B}^\rightarrow$ includes $\ext(a)$,
$h_{B}^\rightarrow$ is unbalanced when $t_C$ is processed. 

The proof
for $h_C^\leftarrow$ is the same as above, replacing
$h_B^\rightarrow$, $h_{B_1}^\rightarrow$, $B$, $B_1$, $t_B$ and $t_C$
with $h_C^\leftarrow$, $h_{C_1}^\leftarrow$, $C$, $C_1$, $t_C$ and
$t_B$, and line~\ref{CheckLR} with line~\ref{CheckRL}.
\end{proof}

\end{document}